\documentclass[12pt]{article}

\usepackage{amsmath}
\usepackage{amssymb}
\usepackage{amsfonts}
\usepackage{latexsym}
\usepackage{color}

\catcode `\@=11 \@addtoreset{equation}{section}
\def\theequation{\arabic{section}.\arabic{equation}}
\catcode `\@=12

\newcommand{\be}{\begin{equation}}
\newcommand{\en}{\end{equation}}
\newcommand{\bea}{\begin{eqnarray}}
\newcommand{\ena}{\end{eqnarray}}
\newcommand{\beano}{\begin{eqnarray*}}
\newcommand{\enano}{\end{eqnarray*}}
\newcommand{\bee}{\begin{enumerate}}
\newcommand{\ene}{\end{enumerate}}

\newcommand{\Rc}{{\cal{R}}}
\newcommand{\mc}{\mathcal}

\newcommand{\D}{{\mc D}}

\newcommand{\V}{{\cal V}}
\newcommand{\LL}{\mc L}

\newcommand{\Sc}{{\cal S}}
\newcommand{\E}{{\cal E}}

\newcommand{\F}{{\cal F}}
\newcommand{\G}{{\cal G}}

\newcommand{\Lc}{{\cal L}}

\newcommand{\C}{{\cal C}}
\newcommand{\1}{1 \!\! 1}

\newcommand{\LD}{{\LL}^\dagger (\D)}

\newcommand{\Hil}{\mc H}
\newcommand{\EE}{\mc E}
\newcommand{\Kc}{\mc K}

\renewcommand{\l}{\langle}
\renewcommand{\r}{\rangle}
\newcommand{\pin}[2]{\l#1 , #2\r}

\newtheorem{thm}{Theorem}

\newtheorem{lemma}[thm]{Lemma}
\newtheorem{prop}[thm]{Proposition}
\newtheorem{defn}[thm]{Definition}

\newenvironment{proof}{\noindent {\bf Proof --}}{\hfill$\square$ \vspace{3mm}\endtrivlist}

\catcode `\@=11 \@addtoreset{equation}{section}
\catcode `\@=12

\begin{document}

\thispagestyle{empty}

\vspace*{2cm}

\begin{center}
{\Large \bf Abstract ladder operators for non self-adjoint Hamiltonians, with applications}   \vspace{2cm}\\

{\large F. Bagarello}\\
  Dipartimento di Ingegneria,
Universit\`a di Palermo,\\ I-90128  Palermo, Italy\\
and I.N.F.N., Sezione di Catania\\
e-mail: fabio.bagarello@unipa.it\\

\end{center}

\vspace*{2cm}

\begin{abstract}
\noindent 
Ladder operators are useful, if not essential, in the analysis of some given physical system since they can be used to find easily eigenvalues and eigenvectors of its Hamiltonian. In this paper we extend our previous results on abstract ladder operators considering in many details what happens if the Hamiltonian of the system is not self-adjoint. Among other results, we give an existence criterion for coherent states constructed as eigenstates of our lowering operators. In the second part of the paper we discuss two different examples of our framework: pseudo-quons and a deformed generalized Heisenberg algebra. Incidentally, and interestingly enough, we show that pseudo-quons can be used to diagonalize an oscillator-like Hamiltonian written in terms of (non self-adjoint) position and momentum operators which obey a deformed commutation rule of the kind often considered in minimal length quantum mechanics.
\end{abstract}

\vspace{2cm}


\vfill


\newpage

\section{Introduction}

In quantum mechanics one of the essential steps to be undertaken is to find eigenvalues and eigenstates of the Hamiltonian $H$ of the physical system one needs to describe, $\Sc$. There are several reasons for this. The first reason is that eigenstates represent the stationary states of $\Sc$, i.e. those states which, if $\Sc$ is {\em prepared} in one of them, maintain $\Sc$ unchanged during its time evolution, if no further (external) action occurs. The second is that the eigenvalues represent the only energy values allowed for $\Sc$. The third is that the set of the eigenstates is, quite often, a basis for $\Hil$, the Hilbert space where $\Sc$ is defined. It is clear, then, that  eigenvalues and eigenstates of $H$ are quite important when dealing with $\Sc$. For this reason, since the birth of quantum mechanics,
different strategies have been proposed to find new solvable Hamiltonians:  interwining operators, \cite{intop1,intop2,intop3},  supersymmetric quantum mechanics, \cite{CKS,jun}, are just two of them. Another well-known startegy makes use of ladder operators, like those appearing in the harmonic oscillator or, in  second quantization and in elementary particles, to deal with bosons and with fermions. There are thousands of books and papers dealing with bosonic and with fermionic operators, and here we only refer to \cite{roman}. Ladder operators also exist in other contexts, like in many models driven by non self-adjoint Hamiltonians, see \cite{mosta,benbook,bagspringerbook,baginbagbook}. In these cases, the point is often that the lowering and the raising operators are not one the adjoint of the other. This creates a lot of freedom, but also many additional, and maybe unexpected, mathematical difficulties. This is possibly the reason why this line of research has become so popular and active in the past few decades, and not only among physicists.

Some years ago, Fernandez started to set up an {\em algebraic treatment} of different quadratic Hamiltonians, not necessarily self-adjoint, \cite{fern1,fern2,fern3}, based on the possibility that, given some Hamiltonian $H$,  one can find an operator $Z$ such that $[H,Z]=\lambda Z$. When this happens, many interesting results can be deduced. In particular, $Z$ turns out to be a ladder operator, and its powers, acting on some {\em seed} eigenstate of $H$, $\hat\varphi$, produces other eigenvectors of $H$, corresponding to different eigenvalues.

In a recent paper, \cite{bagALO}, we have extended Fernandez's results in several ways. In particular, we have considered some classes of {\em abstract ladder operators} (ALOs) useful in different situations, and connected with different Hamiltonians, self-adjoint or not. For each ALO we have shown how to construct eigenvalues and eigenvectors of the related Hamiltonian, and we have proposed some examples arising from pseudo-bosons in one or two dimensions, from quons and from generalized Heisenberg algebra. In particular, we have considered the following situations, all living on a certain Hilbert space $\Hil$, with scalar product $\langle.,.\rangle$, and with related norm $\|.\|=\sqrt{\langle.,.\rangle}$.

\begin{enumerate}
	\item  $H_0$ is a self-adjoint operator acting on $\Hil$, and $Z$ is a second operator on $\Hil$ satisfying the equality
	\be
	[H_0,Z]=\lambda Z,
	\label{11}\en
	for some $\lambda\in\mathbb{R}$. This is essentially Fernandez's case, \cite{fern1,fern2,fern3}.
	
	\item  $H_0$ is again self-adjoint, $H_0=H_0^\dagger$, and $Z$ obeys the following commutation rule:
	\be
	[H_0,Z]=\lambda\, Z [Z^\dagger,Z],
	\label{12}\en
	for some real number $\lambda$. This is what happens, in particular, if $H_0$ is factorizable in terms of $Z$ and $Z^\dagger$, as in $H_0=\lambda Z Z^\dagger$.
	
	\item  $H\neq H^\dagger$ now, and we have a set of operators $Z_j$, $j=1,2,\ldots,N$ such that:
	\be
	[H,Z_j]=\lambda_j Z_j,
	\label{13}\en
	$\lambda_j\in\mathbb{C}$, for $j=1,2,\ldots,N$. This extends (\ref{11}) to several ALOs  and to the case of non self-adjoint $H$.

\end{enumerate}

What is clearly missing, in this list, is the case of (\ref{12}) for an Hamiltonian which is not self-adjoint, $H\neq H^\dagger$. Indeed, this situation was not considered in \cite{bagALO}. In this paper we fill the gap, by considering this particular situation, also in view of its possible applications to the interesting case of pseudo-quons, \cite{bagquons}, and of a particular form of deformed generalized Heisenberg algebra, \cite{bcg}. We will also show how and when it is possible to introduce bi-coherent states in the present settings, \cite{bagspringerbook}, which is a {\em classical} problem to consider in presence of annihilation operators of any kind, \cite{aagbook}-\cite{ABG}.

The paper is organized as follows:

In the next section we will extend formula (\ref{12}) to an Hamiltonian $H$ which is not self-adjoint. We will show that, under some mild assumptions, it is still possible to define ladder operators for $H$ and for $H^\dagger$, and that these operators can be used to construct two families of biorthonormal vectors which are respectively eigenstates of $H$ and of $H^\dagger$. We will also derive the expression of their eigenvalues.

In Section \ref{sect3} we show how to use these families to construct bi-coherent states of the kind considered in \cite{bagspringerbook}, and we give conditions for this to be possible.

In Section \ref{sectDQ} we describe a first class of examples fitting the general construction of our ALOs. This class is based on pseudo-quons, \cite{bagquons}, which are operators obeying a deformed version of the canonical commutation and anti-commutation relations (CCR and CAR) which depends on a parameter $q$, usually taken in the interval $[-1,1]$. Here we will show that $q$ needs not to be in this range, or even to be real. In this perspective our results for quons will be rather general and new, with respect to those existing in the literature so far. We will also connect our pseudo-quons with a special version of the quantum harmonic oscillator, for system with minimal length, and show how these pseudo-quons can be used to diagonalize the (non self-adjoint) Hamiltonian of the system.

Section \ref{sectDGHA} contains another class of examples of ALOs, arising from what has been called {\em deformed generalized Heisenberg algebra}, (DGHA), in \cite{bcg}.

Our conclusions are given in Section \ref{sectconcl}, while a short Appendix is included in the paper to clarify some algebraic aspects of our ALOs, particular useful when these are unbounded operators. A second Appendix concerning graphene is also given as an example of a physical Hamiltonian not bounded from below.

\section{Extending (\ref{12}) to $H\neq H^\dagger$}\label{sect2}

In this section we focus on the possibility of {\em merging} the results deduced in \cite{bagALO} from (\ref{12}) and (\ref{13}). It will not be a surprise to see that this is not so trivial, and requires some effort.

In the following we will deal, most of the time, with three operators $H$, $T$ and $S$. We will always assume to be in one of the following conditions:
\begin{itemize}
	\item[(c1)] $H$, $T$ and $S$ belong to the *-algebra $\Lc^\dagger(\D)$, for some suitable $\D$, see Appendix 1.
	\item[(c2)] it exists a subspace $\D\subseteq\Hil$, dense in $\Hil$, which is stable under the action of $H$, $T$ and $S$ and their adjoints.
\end{itemize}
In both these conditions, we can work with combinations of these operators, their adjoints, and their powers. For instance, $[H,T]=HT-TH$ makes sense as an element of $\Lc^\dagger(\D)$, or simply because when it acts on a vectors $f\in\D$ we get another element of $\D$, just because $HTf=H(Tf)\in\D$ and $THf=T(Hf)\in\D$. Of course, this is not an issue if, say $T$ and $H$ are bounded, since they are defined in (or be extended to) the whole $\Hil$\footnote{As an example of what is meant here that {\em $\D$ is stable} under the action of some given operators, we could consider $c=\frac{1}{\sqrt{2}}\left(x+\frac{d}{dx}\right)$, $c^\dagger=\frac{1}{\sqrt{2}}\left(x-\frac{d}{dx}\right)$ and $H=c^\dagger x$, as for the (shifted) quantum harmonic oscillator. In this case the set $\Sc(\mathbb{R})$ of the test functions ($C^\infty$ functions, decaying to zero, together with all its derivatives, faster than any inverse power of $x$) is stable under the action of each of these operators: $cf(x)\in\Sc(\mathbb{R})$, $c^\dagger f(x)\in\Sc(\mathbb{R})$ and $Hf(x)\in\Sc(\mathbb{R})$ for all $f(x)\in\Sc(\mathbb{R})$. Similar stability can be found in other, and also quite different, situations.}. 

\begin{defn}\label{def1}
	Let $\lambda\in\mathbb{C}$ be a given complex number. We say that $(H,T,S)\in\Rc_\lambda$ if 
	\be
		[H,S]=\lambda S[T,S].
	\label{21}\en
\end{defn}
It is clear that (\ref{21}) extends (\ref{12}), in the sense that (\ref{21}) reduces to (\ref{12}) if $H=H^\dagger$ and if $T=S^\dagger$. In general, in our present settings, there is no reason a priori to require that $\lambda$ is real. As in \cite{bagALO} we will call $T$, $S$, and their adjoints, ALOs.

We have that:

\begin{prop}\label{prop1}
	The following statemes are all equivalent:
	\begin{enumerate}
		\item[(p1)] $[H,S]=\lambda S[T,S]$;
		\item[(p2)] $[H,S^n]=\lambda S[T,S^n]$, $\forall n\geq0$;
		\item[(p3)] $[H^\dagger,S^\dagger]=\overline{\lambda} [T^\dagger,S^\dagger]S^\dagger$;
		\item[(p4)] $[H^\dagger,(S^\dagger)^n]=\overline{\lambda} [T^\dagger,(S^\dagger)^n]S^\dagger$, $\forall n\geq0$.
	\end{enumerate}	
\end{prop}
\begin{proof}
	We only check that from $(p1)$ follows $(p2)$, using induction on $n$. The other statements can be proved easily.
	
	First we notice that $(p2)$ for $n=0$ is trivial, while for $n=1$ $(p2)$ is simply condition $(p1)$. Now, if we suppose that $[H,S^n]=\lambda S[T,S^n]$ is true for a fixed $n$, assuming also $(p1)$, we have
	$$
	[H,S^{n+1}]=[H,S^n]S+S^n[H,S]=\lambda S[T,S^n]S+\lambda S^n(S[T,S])=
	$$
	$$
	=\lambda S\left([T,S^n]S+S^n[T,S]\right)=\lambda S[T,S^{n+1}],
	$$
	which is what we had to prove.
\end{proof}

 From now on we will assume that the {\em lowest} (in modulus, if needed) eigenvalue of $H$ is zero. This is not always a major constraint. In fact, if this is not the case for the given $H'$, meaning that $H'\Phi=\alpha\Phi$, $0<|\alpha|<\infty$, we can still consider $H=H'-\alpha\1$, and $H\Phi=0=0\Phi$, so that $\Phi$ is an eigenstate on $H$ with eigenvalue zero. Moreover, if $(H',T,S)\in\Rc_\lambda$, then $(H,T,S)\in\Rc_\lambda$ as well, since $[H,T]=[H',T]=\lambda S[T,S]$.
 
 \vspace{2mm}
 
 {\bf Remark:--} It might be interesting to stress that such a shift does not always allow to recover what needed here, and in particular that a ground state exists with eigenvalue zero. This is not possible, when $H'$ is not bounded from below (and from above, of course). This is the case, for instance, of Graphene, see \cite{geim} and Appendix 2, but not only, see \cite{kow} for a particle on a circle. 
 
 \vspace{2mm}
 
 A second assumption on $H$ is that it has all eigenvalues with multiplicity one. Also, we will ask the following:
 \begin{equation}
 	[H,[T,S]]=0.
\label{22} \end{equation}

\vspace{2mm}

{\bf Remark:--} We should mention that these conditions are satisfied in many situations. For instance, if $S$ and $T$ are pseudo-bosons, then $[T,S]=\1$, and this commutes with $H$, of course. Also, if $T=a$, $S=b$ and $H=ba$, where $[a,b]_q=ab-qba$ are pseudo-quons, \cite{bagquons}, (\ref{22}) is satisfied. We will return on this particular example in Section \ref{sectDQ}. Also, Theorem \ref{th1} below describes a general case in which (\ref{22}) is automatically satisfied. As for the multiplicity of the eigenvalues, we refer to Sections \ref{sectDQ} and \ref{sectDGHA} for two large class of examples, where this request is satisfied.

\vspace{2mm}

To fix our settings, from now on we will work under the assumption (c2) given at the beginning of this section. We have the following:

\begin{thm}\label{th1}
	Let $(H,S,T)\in\Rc_\lambda$, and let us assume that (\ref{22}) holds and that all the eigenvalues of $H$ are non degenerate. Suppose further that a non zero $\varphi_0\in\D$ exists such that $H\varphi_0=0$. Then, calling
	\be
	\varphi_n=S^n\varphi_0,
	\label{23}\en
	$n\geq0$, and assuming they are all non zero, we put $\F_\varphi=\{\varphi_n, \,\forall n\geq0\}$. Hence we have
	\be
	[T,S]\varphi_n=\mu_n\varphi_n, \qquad \mu_n=\frac{\langle\varphi_n,[T,S]\varphi_n\rangle}{\|\varphi_n\|^2},
	\label{24}\en
	and
	\be
	H\varphi_n=E_n\varphi_n,
	\label{25}\en
	where
	\be
	E_0=0, \qquad E_n=E_{n-1}+\lambda \mu_{n-1}=\lambda\sum_{k=0}^{n-1}\mu_k, \quad n\geq 1.
	\label{26}\en

\end{thm} 

\begin{proof}
	First of all, we observe that (\ref{24}) is a simple consequence of (\ref{22}), and of the fact that all the eigenvalues of $H$ are non degenerate. Hence $[T,S]\varphi_n$, if it is not zero, must be an eigenvector of $H$ with eigenvalue $E_n$ and, therefore, $[T,S]\varphi_n$ must be proportional to $\varphi_n$. We call $\mu_n$ this proportionality constant\footnote{It might be that $[T,S]\varphi_n=0$. In this case, again $[T,S]\varphi_n$ is proportional to $\varphi_n$, with $\mu_n=0$. Again, (\ref{24}) is satisfied, but it is less interesting. We will assume all throughout this paper $\mu_n\neq0$ for all $n\geq0$.}. Now, taking the scalar product of this equality with $\varphi_n$, and using the fact that $\varphi_n\neq0$, we conclude the proof of (\ref{24}).

	As for (\ref{25}) and (\ref{26}), we use induction on $n$, starting with $n=0$, which is obviously true.

	Let us now assume that (\ref{25}) and (\ref{26}) hold for some given $n$. We want to prove that the same formulas hold if we replace $n$ with $n+1$. For that we write
	$$
	H\varphi_{n+1}=HS\varphi_n=([H,S]+SH)\varphi_n=(\lambda S[T,S]+SH)\varphi_n=(\lambda \mu_n+E_n)S\varphi_n,
	$$
	which is exactly formula (\ref{25}) with $E_{n+1}=E_n+\lambda \mu_n$, as in (\ref{26}).

\end{proof}

\vspace{2mm}

{\bf Remark:--} First we observe that the stability of $\D$ under the action of, say, $S$, implies that $\varphi_n\in\D$ for all $n\geq0$. But $\D$ is also stable under the action of $H$ and $T$, and their adjoints, so that $[S,H]:\D\rightarrow\D$. In principle it could happen that, for some given $n_0\in\mathbb{N}$, $S^{n_0}\varphi_0\neq0$ while $S^{n_0+1}\varphi_0=0$. In this case, we could still set up a strategy similar to the one we will describe here, but with some changes. However, we will not consider this case here, focusing only on the case in which each $\varphi_n$ in (\ref{23}) is different from zero.

\vspace{2mm}

Due to the fact that $H\neq H^\dagger$, it is natural to ask if the same results above, or similar, can be restated for $H^\dagger$. This indeed can be done if we work under the assumptions of Theorem \ref{th1}, and if we further assume that $\F_\varphi$ is a basis for $\Hil$. In this case, in fact, an unique other basis of $\Hil$, $\F_\psi=\{\psi_n,\, n\geq0\}$, exists which is biorthonormal to $\F_\varphi$:
\begin{equation}
	\langle\varphi_n,\psi_m\rangle=\delta_{n,m},
\end{equation}
for all possible $n$ and $m$. We refer to \cite{chri,heil} for this result. In particular, any $f\in\Hil$ can be expanded as follows:
\be
f=\sum_{n=0}^\infty\langle \varphi_n,f\rangle\,\psi_n=\sum_{n=0}^\infty\langle \psi_n,f\rangle\,\varphi_n.
\label{27}\en
What is also quite relevant for us is that, using the completeness of $\F_\varphi$, it is possible to prove, using standard ideas, see e.g. \cite{bagspringerbook}, that
\be
H^\dagger \psi_n=\overline{E_n}\psi_n,
\label{28}\en
and
\begin{equation}
	S^\dagger \psi_0=0, \qquad S^\dagger \psi_n=\psi_{n-1}, \quad n\geq1.
\label{29}\end{equation}
Hence we see that, while $S$ acts as a raising operator for $\F_\varphi$, its adjoint, $S^\dagger$, behaves as a lowering operator on $\F_\psi$. We also observe that the eigenvalues of $H^\dagger$, as those of $H$, are non degenerate. These results are similar to others we have deduced in different situations, again in presence of some non self-adjoint Hamiltonian. As for equation (\ref{24}), we can prove the following:

\begin{lemma}\label{lemma1}
	The equations in (\ref{24}) are equivalent to 
	\be
		[T^\dagger,S^\dagger]\psi_n=-\overline{\mu_n}\psi_n,
	\label{210}\en
	for all $n\geq0$.
\end{lemma}
\begin{proof}
	Let us take $n,m\geq0$, and let us consider
	$$
	\langle[T^\dagger,S^\dagger]\psi_n,\varphi_m\rangle=-\langle\psi_n,[T,S]\varphi_m\rangle=-\mu_m\langle\psi_n,\varphi_m\rangle=-\mu_m\delta_{n,m}=
	$$
	$$
	=-\mu_n\langle\psi_n,\varphi_m\rangle=-\langle\overline{\mu_n}\,\psi_n,\varphi_m\rangle,
	$$
	where we have used (\ref{24}). Now, for each fixed $n$, we have, $\langle([T^\dagger,S^\dagger]-\overline{\mu_n})\psi_n,\varphi_m\rangle=0$, $\forall m\geq0$. Formula (\ref{210}) follows then from the completeness of $\F_\varphi$.
	
\end{proof}

Incidentally we observe that, comparing (\ref{210}) and (\ref{24}), we have the following identity:
\be
\langle\hat\varphi_n,[T,S]\hat\varphi_n\rangle=\langle\hat\psi_n,[T,S]\hat\psi_n\rangle,
\label{212}\en
$\forall n\neq0$. Here $\hat\varphi_n=\frac{\varphi_n}{\|\varphi_n\|}$, and $\hat\psi_n=\frac{\psi_n}{\|\psi_n\|}$.

\vspace{2mm}

Going now back to Proposition \ref{prop1}, we observe that all the commutators considered there in the left-hand sides involve $H$, $S$, and their adjoints. The operator $T$, in any of its forms, only appears in the right-hand sides of $(p1)-(p4)$. Hence, it is interesting to consider commutators like $[H,T]$, $[H^\dagger,T^\dagger]$, and so on. This is also because, if $H$ can be factorized as $H=\lambda ST$, then the fact that $(p1)$ is satisfied is trivial. But, in the very same way, we also deduce that $[H,T]=\lambda[S,T]T$, and $[H^\dagger,T^\dagger]=\overline{\lambda} T^\dagger[S^\dagger,T^\dagger]$. This suggests to refine Definition \ref{def1} as follows:
\begin{defn}\label{def2}
	A triple $(H,T,S)\in \Rc_\lambda$ is in  ${\Rc}_\lambda^{(s)}$ if 
	\begin{equation}
		[H,T]=\lambda[S,T]T.
	\label{213}\end{equation}
\end{defn}
Here the suffix $(s)$ stands for {\em strong}. It follows that $(H,T,S)\in {\Rc}_\lambda^{(s)}$ if and only if, for instance, 
\begin{equation}
	[H^\dagger,S^\dagger]=\overline{\lambda} [T^\dagger,S^\dagger]S^\dagger, \qquad\mbox{and}\qquad [H^\dagger,T^\dagger]=\overline{\lambda} T^\dagger[S^\dagger,T^\dagger].
\label{214}\end{equation}
We will now show that, if Definition \ref{def2} is satisfied,  $T^\dagger$ acts as a raising operator for $\F_\psi$, while its adjoint $T$ acts as a lowering operator for $\F_\varphi$. More explicitly we will now show that a sequence of complex numbers $\{\gamma_n\}$ can be found such that
\be
T^\dagger\psi_n=\gamma_n\psi_{n+1},
\label{215}\en
$n\geq0$, and
\be
T\varphi_0=0, \qquad T\varphi_n=\overline{\gamma_{n-1}}\varphi_{n-1}, \quad n\geq1.
\label{216}\en

To prove (\ref{215}) we first observe that $T^\dagger H^\dagger\psi_n=\overline{E_n}T^\dagger\psi_n$, because of (\ref{28}). Also, because of (\ref{214}) and of (\ref{210}), we have
$$
[H^\dagger,T^\dagger]\psi_n=\overline{\lambda} T^\dagger[S^\dagger,T^\dagger]\psi_n=\overline{\lambda} T^\dagger\left(\overline{\mu_{n}}\psi_{n}\right)=\overline{\lambda\,\mu_{n}} (T^\dagger\psi_{n}).
$$ 
Hence we have, using (\ref{26}),
$$
H^\dagger\left(T^\dagger\psi_n\right)=\left([H^\dagger,T^\dagger]+T^\dagger H^\dagger\right)\psi_n=\overline{E_{n+1}}T^\dagger\psi_{n}.
$$
This means that $T^\dagger\psi_{n}$ is an eigenstate of the non-degenerate eigenvalue of $H^\dagger$, $\overline{E_{n+1}}$, so that it must be proportional to $\psi_{n+1}$, as in (\ref{215}).

As for (\ref{216}), this can be proved with a similar technique as in Lemma \ref{lemma1}, using the completeness of $\F_\psi$.

From the ladder equations deduced so far we also find that
\be
ST\varphi_n=\overline{\gamma_{n-1}}\varphi_{n}, \qquad  TS\varphi_n=\overline{\gamma_{n}}\varphi_{n},
\label{217}\en
and
\be
S^\dagger T^\dagger\psi_n={\gamma_{n}}\psi_{n}, \qquad  T^\dagger S^\dagger\psi_n={\gamma_{n-1}}\varphi_{n},
\label{218}\en
$n\geq0$, with the agreement that $\gamma_{-1}=0$. Of course these results are coherent with (\ref{22}), and with the analogous for the adjoint operators. In fact we find easily that 
\be
[H,ST]=[H,TS]=[H^\dagger,S^\dagger T^\dagger]=[H^\dagger,T^\dagger S^\dagger]=0
\label{219}\en
Moreover, we deduce that
\be
[S,T]\varphi_n=\left(\overline{\gamma_{n-1}}-\overline{\gamma_{n}}\right)\varphi_n, \qquad [S^\dagger,T^\dagger]\psi_n=\left(\gamma_{n}-\gamma_{n-1}\right)\psi_n,
\label{220}\en
which can be rewritten, in a bra-ket language,
as
\be
[S,T]=\sum_{n=0}^\infty\left(\overline{\gamma_{n-1}}-\overline{\gamma_{n}}\right)|\varphi_n\rangle\langle\psi_n|, \qquad [S^\dagger,T^\dagger]=\sum_{n=0}^\infty\left(\gamma_{n}-\gamma_{n-1}\right)|\psi_n\rangle\langle\varphi_n|.
\label{221}\en
If $\gamma_n-\gamma_{n-1}=\Gamma$, independent of $n$, then we conclude that, since $\F_\varphi$ and $\F_\psi$ are biorthogonal bases, $[S,T]=-\overline{\Gamma}\,\1$, and $[S^\dagger,T^\dagger]=\Gamma\,\1$.

\section{Introducing bi-coherent states for our ALOs}\label{sect3}

So far, we have no particular reason to fix $\{\gamma_n\}$: this is just a complex-valued sequence arising from our discussion above, see (\ref{215}). In what follows we will show how the various $\gamma_n$'s must be chosen in order to introduce (well-defined) vectors in $\Hil$ which are eigenstates of the lowering operators $T$ and $S^\dagger$ considered in our settings. In other words, we would like to construct, following the same general procedure proposed in \cite{bagspringerbook}, two vectors $\varphi(z), \psi(z)\in\Hil$,  $z\in \E$, such that
\be
T\varphi(z)=z\varphi(z), \qquad S^\dagger\psi(z)=z\psi(z).
\label{31}\en
Here $\E$ is some {\em sufficiently large} subset of $\mathbb{C}$, to be identified. As usual, we look for solutions of (\ref{31}) of the following type:
\be
\varphi(z)=N_\varphi(z)\sum_{n=0}^\infty\alpha_n z^n\varphi_n, \qquad \psi(z)=N_\psi(z)\sum_{n=0}^\infty\beta_n z^n\psi_n,
\label{32}\en
where $N_\varphi(z)$ and $N_\psi(z)$ are some $z$-dependent normalization for the states, while $\{\alpha_n\}$ and $\{\beta_n\}$ are complex-valued sequences to be fixed. 

Using (\ref{216}) with simple computations we find that
$$
T\varphi(z)=z\left(N_\varphi(z)\sum_{n=0}^\infty\alpha_{n+1}\overline{\gamma_n} z^n\varphi_n\right),
$$
which is equal to $z\varphi(z)$ if the following relation is satisfied $\forall n\geq0$:
\be
\alpha_{n+1}\overline{\gamma_n}=\alpha_n.
\label{33}\en
If $\gamma_0=0$, then (\ref{33}) implies that $\alpha_0=0$. Moreover, from (\ref{216}) we deduce that $T\varphi_0=T\varphi_1=0$, while from (\ref{215}) we find also that $T^\dagger \psi_0=0$, which should be added to the annihilation rule in (\ref{29}), $S^\dagger \psi_0=0$. These two facts are unusual ($T$ has two vacua, and $\psi_0$ is annihilated by two, in general different, operators, $T^\dagger$ and $S^\dagger$). To avoid these situations, which however could have some interest and could be considered anyhow\footnote{However, we will not do it here.}, we suppose that $\gamma_n\neq0$ for all $n\geq0$. When this is the case, we deduce that
\be
\alpha_n=\frac{\alpha_0}{\overline{\gamma_{n-1}!}},
\label{34}\en
$\forall n\geq1$. Here we have defined $\gamma_k!=\gamma_0\gamma_1\cdots\gamma_k$, $k\geq1$. In the rest of the paper we will define $\gamma_{-1}!=1$, and we take $\alpha_0=1$. This is not restrictive, since the normalization of $\varphi(z)$ is still to be fixed in some way by $N_\varphi(z)$. Summarizing we get
\be
\varphi(z)=N_\varphi(z)\sum_{n=0}^\infty \frac{z^n}{\overline{\gamma_{n-1}!}}\varphi_n,
\label{35}\en
for all $z\in\E$, still to be identified.

In a similar way, from $S^\dagger\psi(z)=z\psi(z)$, we find the following expression for $\psi(z)$:
\be
  \psi(z)=N_\psi(z)\sum_{n=0}^\infty z^n\psi_n,
\label{36}\en
since we can deduce that $\beta_0=\beta_1=\beta_2=\ldots$, in (\ref{32}), and we are fixing $\beta_0=1$. Once we have found these formal expressions for $\varphi(z)$ and $\psi(z)$ we want to make these formulas rigorous. More explicitly, we want to find conditions for the above series to converge in some region of $\mathbb{C}$. We proceed as in \cite{bagspringerbook}, and references therein, giving a (rather mild) sufficient condition which ensures tha convergence of the series in (\ref{35}) and (\ref{36}). For that we assume that four positive constant exist, $A_\varphi, A_\psi, r_\varphi, r_\psi$, and two strictly positive sequences $\{M_n(\varphi)\}$ and $\{M_n(\psi)\}$, such that
\be
\lim_{n,\infty}\frac{M_n(\varphi)}{M_{n+1}(\varphi)}=M(\varphi), \qquad \lim_{n,\infty}\frac{M_n(\psi)}{M_{n+1}(\psi)}=M(\psi),
\label{37}\en
with $M(\varphi)>0$ and $M(\psi)>0$ satisfying the following inequalities:
\be
\|\varphi_n\|\leq A_\varphi r_\varphi^n M_n(\varphi), \qquad \|\psi_n\|\leq A_\psi r_\psi^n M_n(\psi),
\label{38}\en
$\forall n$. Then we have
$$
\left\|\sum_{n=0}^\infty \frac{z^n}{\overline{\gamma_{n-1}!}}\varphi_n\right\|\leq \sum_{n=0}^\infty \frac{|z|^n}{|\gamma_{n-1}!|}\|\varphi_n\|\leq A_\varphi \sum_{n=0}^\infty \frac{(r_\varphi|z|)^n}{|\gamma_{n-1}!|}\,M_n(\varphi),
$$
which is a power series in $(r_\varphi|z|)$. Calling $\gamma=\lim_n|\gamma_n|$, we can easily check that the series converge whenever $|z|<\frac{\gamma M(\varphi)}{r_\varphi}$, which could be all of $\mathbb{C}$ if $M(\varphi)=\infty$ or $\gamma=\infty$. Similarly, the series $\sum_{n=0}^\infty z^n\psi_n$ converges when $|z|<\frac{M(\psi)}{r_\psi}$. To conclude that $\varphi(z)$ and $\psi(z)$ are well defined (in the common region of $\mathbb{C}$ in which both bounds on $z$ are satisfied), we need to understand how $N_\varphi(z)$ and $N_\psi(z)$ must be defined. As usual, \cite{bagspringerbook}, we require that $\varphi(z)$ and $\psi(z)$ satisfy a sort of {\em mutual normalization}:
$$
1=\langle\varphi(z),\psi(z)\rangle=\overline{N_\varphi(z)}\,N_\psi(z)\sum_{n,k=0}^\infty \frac{\overline z^n z^k}{\gamma_{n-1}!}\langle\varphi_n,\psi_k\rangle=\overline{N_\varphi(z)}\,N_\psi(z)\,\Gamma(z),
$$
where we have introduced
\be\Gamma(z)=\sum_{n=0}^\infty\frac{|z|^{2n}}{\gamma_{n-1}!},
\label{39}\en
which converges if $|z|<\sqrt{\gamma}$. Of course, for all those $z$ satisfying this inequality and for which $\Gamma(z)\neq0$, we conclude that
\be
\overline{N_\varphi(z)}\,N_\psi(z)=\frac{1}{\Gamma(z)}.
\label{310}\en
{\bf Remark:--} We should stress once more that, since $\gamma_n$ is not necessarily real, $\Gamma(z)$ is not necessarily a real function, even if it only depends on $|z|$ (or, to be even more explicit, on $|z|^2$).
\vspace{2mm}
We are now in a position to describe what $\E$ must be: indeed we have $\E=C_\rho(0)$, the circle centered in the origin and of radius $\rho$, where
$$
\rho=\min\left\{\frac{\gamma M(\varphi)}{r_\varphi},\frac{M(\psi)}{r_\psi},\sqrt{\gamma}\right\}.
$$
Now, taken $f,g\in\Hil$, we have
$$
\int_{C_\rho(0)}d\nu(z,\overline z)\,\langle f,\psi(z)\rangle\langle \varphi(z),g\rangle=\sum_{n,m=0}^\infty \frac{1}{\overline{\gamma_{n-1}!}}\,\langle f,\psi_n\rangle\langle \varphi_m,g\rangle\int_{C_\rho(0)}d\nu(z,\overline z)\,\frac{z^n\,\overline z^m}{\Gamma(z)}.
$$
If we now put $d\nu(z,\overline z)=\Gamma(z)\,d\lambda(r)\,d\theta$, $r\in[0,\rho[$ and $\theta\in[0,2\pi[$, and we integrate out the angular part, we conclude that
$$
\int_{C_\rho(0)}d\nu(z,\overline z)\,\langle f,\psi(z)\rangle\langle \varphi(z),g\rangle=2\pi\sum_{n=0}^\infty \frac{1}{\overline{\gamma_{n-1}!}}\,\langle f,\psi_n\rangle\langle \varphi_n,g\rangle\,\int_0^\rho\,d\lambda(r)r^{2n}=$$
$$=\sum_{n=0}^\infty\langle f,\psi_n\rangle\langle \varphi_n,g\rangle=\langle f,g\rangle,
$$
since $\F_\varphi$ and $\F_\psi$ are biorthogonal bases. However, this is possible only if we can find a {\em generalized measure} $d\lambda(r)$ such that
$$
\int_0^\rho\,d\lambda(r)r^{2n}=\frac{1}{2\pi}\,\overline{\gamma_{n-1}!},
$$
$\forall n\geq0$. In particular, this implies that $\int_0^\rho\,d\lambda(r)=0$. This already clarify that   $d\lambda(r)$ cannot be a measure in the usual sense. This is also clear since the integral $\int_0^\rho\,d\lambda(r)r^{2n}$, for $n\geq1$, can give complex results.

Summarizing, $d\nu(z,\overline z)$ must be chosen properly, and there is no general reason ensuring that this can be done. Hence, resolution of the identity for the bi-coherent states is an open problem, while their proper definition, and the fact that they are eigenstates of $T$ and $S^\dagger$, is granted for all $z\in\C_\rho(0)$.

\section{Pseudo-quons}\label{sectDQ}
In this section we will discuss a concrete example fitting well what we have discussed in Section \ref{sect2}. This example is based on pseudo-quons, \cite{bagspringerbook,bagquons}, but in a slightly revised version, where the $q$-parameter, usually restricted in the real range $[-1,1]$, can assume complex values as well. This possibility, in our knowledge,  was not considered in the literature before, and open new possibilities.

More explicitly, {\em ordinary} quons arise from an operator $c$ which, together with its adjoint $c^\dagger$, satisfies the following $q$-mutator rule: $cc^\dagger-qc^\dagger c=\1$, \cite{fiv,green}. This implies, taking the adjoint of both members of this equation, that $cc^\dagger-\overline q c^\dagger c=\1$ as well, so that $(\overline q-q)c^\dagger c=0$, which is possible only if $q\in\mathbb{R}$.

As in \cite{bagspringerbook,bagquons}, we rather consider a deformed version of the $q$-mutator:
\be
[a,b]_q=ab-qba=\1,
\label{41}\en
where $a$ is, in general, different from $b^\dagger$. As already discussed in Section \ref{sect2}, see also \cite{bagquons}, formula (\ref{41}) should be understood in general in the sense of unbounded operators. This means that we should have either a common dense domain of $\Hil$, $\D$, stable under the action of $a$, $b$, and their adjoints, or an algebraic settings where these operators naturally {\em live}, as the *-algebra $\Lc^\dagger(\D)$, see  Appendix 1. We work here in one of these conditions. 

It is interesting to notice that, going from quons to pseudo-quons, allows us to deal also with complex values of $q$. This is what we will do in the first part of this section. More explicitly, we will deduce some results arising from (\ref{41}) under the general assumption that $q\in\mathbb{C}$. To do this, we will follow essentially the same ideas described in \cite{bagquons}. In fact, most of what we are going to show here is a simple extension of the results in \cite{bagquons}.

We first assume that two non zero vectors $\Phi_0,\Psi_0$ exist in $\D$  such that
\be
a\Phi_0=b^\dagger\Psi_0=0.
\label{42}\en
Next we define
\be
\Phi_n=\alpha_n\, b\,\Phi_{n-1}=\alpha_n!\,b^n\,\Phi_0, \qquad \mbox{and}\qquad \Psi_n=\beta_n\, a^\dagger\,\Psi_{n-1}=\beta_n!\,{a^\dagger}^n\,\Psi_0,
\label{43}\en
for $n\geq1$, and the sets $\F_\Phi=\{\Phi_n,\,n\geq0\}$ and $\F_\Psi=\{\Psi_n,\,n\geq0\}$. We have introduced here $\alpha_n!=\alpha_1\alpha_2\cdots\alpha_n$ and $\beta_n!=\beta_1\beta_2\cdots\beta_n$, for all $n\geq1$. This is slightly different from what we have done in (\ref{34}), where the factorial included also $\gamma_0$. Notice that these quantities could be complex, in principle. We further put $\alpha_0!=\gamma_0!=1$. We introduce
\be
[n]_q=1+q+q^2+\ldots+q^{n-1}=\frac{1-q^n}{1-q},
\label{44}\en
with the agreement that $[0]_q=0$ (which, by the way, is what we get from the right-hand side of this formula). Then, putting $N=ba$, we can check that
\be
N\Phi_n=[n]_q\Phi_n, \qquad N^\dagger\Psi_n=\overline{[n]_q}\,\Psi_n,
\label{45}\en
for all $n\geq0$. The proof is similar to that in \cite{bagquons} and will not be given here. A standard consequence of (\ref{45}) is that
\be
\langle\Psi_n,\Phi_m\rangle=0, \qquad \forall n\neq m.
\label{46}\en
This is because, if $q\neq1$, $[n]_q=[m]_q$ if and only if $n=m$. Formula (\ref{46}) can be refined requiring first that the two vacua $\Psi_0$ and $\Phi_0$ satisfy the following:
\be
\langle\Psi_0,\Phi_0\rangle=1.
\label{47}\en
If we further assume that
\be
\overline{\beta_n}\,\alpha_n\,[n]_q=1,
\label{48}\en
for all $n\geq1$, we conclude that
\be
\langle\Psi_n,\Phi_m\rangle=\delta_{n,m},
\label{49}\en
$\forall n,m\geq0$. Of course, due to (\ref{46}), we only have to prove that $\langle\Psi_n,\Phi_n\rangle=1$ for all $n\geq0$. This is true (by construction) for $n=0$. Now, let us assume that $\langle\Psi_n,\Phi_n\rangle=1$ for a fixed $n$. Then we have
$$
\langle\Psi_{n+1},\Phi_{n+1}\rangle=\overline{\beta_{n+1}}\,\alpha_{n+1}\langle a^\dagger\Psi_n,b\,\Phi_n\rangle=\overline{\beta_{n+1}}\,\alpha_{n+1}\langle \Psi_n,a\,b\,\Phi_n\rangle=\overline{\beta_{n+1}}\,\alpha_{n+1}\langle \Psi_n,(\1+qN)\,\Phi_n\rangle.
$$
Now, using (\ref{45}) and the induction assumption, we have
$$
\langle\Psi_{n+1},\Phi_{n+1}\rangle=\overline{\beta_{n+1}}\,\alpha_{n+1}\left(1+q[n]_q\right)=\overline{\beta_{n+1}}\,\alpha_{n+1} [n+1]_q,
$$
so that from (\ref{48}) we conclude that $\langle\Psi_{n+1},\Phi_{n+1}\rangle=1$, as we had to prove. Hence we conclude that  $\F_\Phi$ and $\F_\Psi$ are biorthonormal sets. It is a simple exercise to prove further that
\be
\left\{
\begin{array}{ll}
a\,\Phi_0=0, \hspace{1cm} a\,\Phi_n=\alpha_n\,[n]_q
\Phi_{n-1}, \hspace{.8cm} n\geq1\\
b\,\Phi_n=\frac{1}{\alpha_{n+1}}\,\Phi_{n+1}, \hspace{3.7cm} n\geq0\\
b^\dagger\,\Psi_0=0, \hspace{1cm} b^\dagger\,\Psi_n=\gamma_n\,\overline{[n]_q}\,
\Psi_{n-1}, \quad n\geq1\\
a^\dagger\,\Psi_n=\frac{1}{\gamma_{n+1}}\,\Psi_{n+1}, \hspace{3.7cm} n\geq0.
\end{array}
\right.
\label{410}\en 
Notice that (\ref{45}) can also be deduced from these ladder equations. As in \cite{bagspringerbook} we could assume, if needed\footnote{Of course, this should be checked in concrete situations, since it is not automatic.}, that  $\F_\Phi$ and $\F_\Psi$ are $\G$-quasi bases:
\be
\langle f,g\rangle=\sum_n\langle f,\Phi_n\rangle\langle \Psi_n,g\rangle=\sum_n\langle f,\Psi_n\rangle\langle \Phi_n,g\rangle,
\label{411}\en
$\forall f,g\in\G$, a suitable dense subspace of $\Hil$. In \cite{bagspringerbook,baginbagbook} it is widely discussed the consequences of this property, which generalizes the Parceval identity for orthonormal bases, and many concrete physical systems which admit $\G$-quasi bases of eigenvectors of some Hamiltonian are discussed. For completeness, we give here an example of how the operators $a$ and $b$ look like in concrete cases, referring to \cite{bagquons} for more details (including the spaces where these operators act, and so on):
\be
a=\frac{e^{-2i\alpha x}-e^{i\alpha \frac{d}{dx}}e^{-i\alpha (x+\gamma)}}{-i\sqrt{1-e^{-2\alpha^2}}}, \quad b=\frac{e^{2i\alpha x}-e^{i\alpha (x-\gamma)}e^{i\alpha \frac{d}{dx}}}{i\sqrt{1-e^{-2\alpha^2}}},
\label{ex1}\en
for real $\gamma\neq0$. For $\gamma=0$ these operators return a pair $(c,c^\dagger)$ of ordinary quons, \cite{eremel}. We recall that, in (\ref{ex1}), $q=e^{-2\alpha^2}$ where $\alpha\in[0,\infty)$, so that $q\in]0,1]$.

\subsection{Pseudo-quons in a deformed oscillator, and their role for ALOs}

In this section we use the operators $a$ and $b$ in (\ref{41}) to construct, first, a quantum oscillator-like non self-adjoint Hamiltonian $H$ written in terms of two (again, non self-adjoint) operators which extend the position and the momentum operators, but which obey non standard commutation relations, of the kind we can find in \cite{nozari,fring2012,fring} and, more recently, in \cite{samar}. One interesting aspect of our construction is that, because of what we have seen  before in Section \ref{sectDQ}, we will be able to find eigenvalues and eigenvectors of $H$ and $H^\dagger$. Secondly we will show how $a$, $b$, $H$ and $H^\dagger$ fit the abstract settings of Section \ref{sect2}.

First of all we use  $a$ and $b$ in (\ref{41}) to define two operators $x$ and $p$ as follows:
\be
x=\alpha(b+a), \qquad p=i\beta(b-a),
\label{412}\en
where $\alpha$ and $\beta$ are, in general, complex quantities. Recalling that $N=ba$ we can rewrite
\be
[x,p]=2i\alpha\beta[a,b]=2i\alpha\beta\left(\1+(q-1)N\right).
\label{413}\en
Now, since $a=\frac{\beta x+i\alpha p}{2\alpha\beta}$ and $b=\frac{\beta x-i\alpha p}{2\alpha\beta}$ we have
$$
N=\frac{1}{4\alpha^2\beta^2}\left(\beta^2 x^2+\alpha^2p^2+i\alpha\beta[x,p]\right),
$$
which, when replaced in the right-hand side of (\ref{413}), returns the following commutator between $x$ and $p$:
\be
[x,p]=\frac{4i\alpha\beta}{q+1}\1+\frac{i(q-1)}{\alpha\beta(q+1)}\left(\beta^2x^2+\alpha^2p^2\right).
\label{414}\en
This formula, of course, makes sense if $q\neq-1$. Going back to (\ref{41}) this means that we cannot include fermions in our analysis here. Incidentally we observe that, if $q=1$ (i.e. for ordinary bosons), we go back to the well known commutation rule $[x,p]=i\1$ if we fix, as for ordinary bosons, $\alpha=\beta=\frac{1}{\sqrt{2}}$ in (\ref{412}). We refer to \cite{nozari,fring2012,fring,samar}, and references therein, for some results on commutation rules similar to the one in (\ref{414}), and for their role in  minimal length quantum mechanics. Here, we rather  introduce the operator
\be
H=\frac{1}{2}\left(\left(\frac{p}{\beta}\right)^2+\left(\frac{x}{\alpha}\right)^2\right),
\label{415}\en
which looks formally as a quantum oscillator in the rescaled variables $\frac{p}{\beta}$ and $\frac{x}{\alpha}$. Of course this is not really so since $\alpha$ and $\beta$ are, in general, complex, and $p$ and $x$ are not self-adjoint. This is also made more explicit by noticing that, using (\ref{413}) and (\ref{414}), we can rewrite
\be
H=(q+1)N+\1, \qquad \mbox{and} \qquad H^\dagger=(\overline q+1)N^\dagger+\1,
\label{416}\en
which show that the families $\F_\Phi$ and $\F_\Psi$ introduced after (\ref{43}) are indeed the eigenstates of $H$ and $H^\dagger$, respectively:
\be
H\Phi_n=\hat E_n\Phi_n,  \qquad \mbox{and} \qquad H^\dagger\Psi_n=\overline{\hat E_n}\Psi_n,
\label{417}\en
$\forall n\geq0$. Here we have
\be
\hat E_n=(q+1)[n]_q+1,
\label{418}\en
which is not necessarily real, of course. We can rewrite $\hat E_n$ as follows:
\be
\hat E_0=1, \qquad\mbox{and}\qquad \hat E_n=2[n]_q+q^n, \quad n\geq1.
\label{419}\en
Hence our pseudo-quons can be used to diagonalize both $H$ in (\ref{415}) and $H^\dagger$, where $x$ and $p$ are (non self-adjoint) operators obeying (\ref{414}).

\vspace{2mm}

What we want to do now is to show that $a$, $b$ and $H$ produce an explicit realization of what discussed in Section \ref{sect2}. For that we put, in Definition \ref{def1}, $S=b$ and $T=a$, while $H$ is the one in (\ref{416}). Hence we have, as already commented before,
$$
[H,b]=(q+1)[N,b]=(q+1)b[a,b],
$$
which is exactly (\ref{21}) with $\lambda=q+1$. Hence $(H,a,b)\in\Rc_{q+1}.$ However, in view of Theorem \ref{th1}, this $H$ does not work well, since its lowest eigenvalue is one, rather than zero. But this, as we have already observed before, is not really a major problem. It is sufficient to replace $H$ with $h=H-\1=(q+1)N$. Hence $h$ is a shifted version of (\ref{415}), and its eigenstates, and those of $h^\dagger$, are not modified by this shift, while the eigenvalues are slightly changed. Summarizing, we have
\be
h=(q+1)N, \qquad \mbox{and} \qquad h^\dagger=(\overline q+1)N^\dagger,
\label{420}\en
\be
h\Phi_n=E_n\Phi_n,  \qquad \mbox{and} \qquad h^\dagger\Psi_n=\overline{E_n}\Psi_n,
\label{421}\en
$\forall n\geq0$, with
\be
E_n=(q+1)[n]_q, \qquad n\geq0.
\label{422}\en
Since $[h,b]=(q+1)[N,b]=(q+1)b[a,b]$, we go back again to (\ref{21}) with $\lambda=q+1$, so that $(h,a,b)\in\Rc_{q+1}$. This means that Definition \ref{def1} is satisfied, and Proposition \ref{prop1} easily follows, replacing $(H,S,T,\lambda)$ with $(h,b,a,q+1)$.

The next steps are meant to show what the general settings in Section \ref{sect2} become in this particular case. We start with showing that (\ref{22}) is satisfied. Indeed we have
$$
[H,[T,S]]\rightarrow [h,[a,b]]=[(q+1)N,\1+(q-1)N]=0.
$$
Here we have used the following useful identity:
\be
[a,b]=ab-ba=(\1+qN)-N=\1+(q-1)N.
\label{423}\en
This means that we are under the assumptions of Theorem \ref{th1}: in fact $(h,a,b)\in\Rc_{q+1}$, (\ref{22}) is satisfied and all the eigenvalues of $h$, see (\ref{422}), are non degenerate. Moreover, $h\Phi_0=0$ so that the vector $\varphi_0$ in Theorem \ref{th1} exists, and coincides with the vacuum of $a$, see (\ref{42}). Next, as in (\ref{23}) we put $\varphi_n=S^n\varphi_0=b^n\Phi_0$. Notice that these vectors are proportional to the $\Phi_n$'s in (\ref{43}): $\varphi_n=(\alpha_n!)^{-1}\,\Phi_n$, $n\geq0$. We easily find that $\mu_n$ in (\ref{24}) and $E_n$ in (\ref{26}) are
\be
\mu_n=1+(q-1)[n]_q=q^n, \qquad E_n=(q+1)[n]_q,
\label{424}\en
$\forall n\geq0$. This last expression, in particular, is in agreement with (\ref{422}).

In Section  \ref{sect2} the family $\F_\psi$ has been introduced as the (only) biorthonormal basis of $\F_\varphi$. For pseudo-quons, the biorthonormality condition is given in (\ref{49}). This allows us to identify the $\psi_n$ in (\ref{27}) in terms of the $\Psi_n$ in (\ref{43}):
\be
\Phi_n=\alpha_n!\varphi_n, \qquad \mbox{and}\qquad \Psi_n=(\overline{\alpha_n!})^{-1}\psi_n,
\label{425}\en
$\forall n\geq0$. Incidentally we observe that, because of the (\ref{48}), both $\alpha_n$ and $\beta_n$ must be non zero for all $n$. In what follows it will be useful to use the following equality, which is a consequence, in particular, of (\ref{425}): 
\be
\psi_n=[(\overline{\alpha_n}\,\beta_n)!]\,(a^\dagger)^n\psi_0=[\overline{\alpha_n}\,\beta_n]\,a^\dagger\,\psi_{n-1},
\label{426}\en
$n\geq1$, where we are identifying $\psi_0$ and $\Psi_0$, since $\alpha_0!=1$. Formulas (\ref{28}) and (\ref{29}) can be easily checked and we get indeed
$$
h^\dagger\psi_n=(\overline{q}+1)\overline{[n]_q}\psi_n, \qquad b^\dagger\psi_0=0, \qquad \mbox{and, if $n\geq1$, }\qquad b^\dagger\psi_n=\psi_{n-1}.
$$
The two mean values of $[T,S]=[a,b]$ in (\ref{212}) indeed coincide and return $q^n$, $\forall n\geq0$. 

What is maybe more interesting is the fact that $(h,a,b)\in{\Rc}_{q+1}^{(s)}$, too. This is because (\ref{213}) is also satisfied: $[h,a]=(q+1)[b,a]a$. For this reason $T^\dagger$ and $T$ (i.e. $a^\dagger$ and $a$) act respectively as raising and lowering operators on $\psi_n$ and $\varphi_n$, see (\ref{215}) and (\ref{216}). In particular, it turns out that the sequence $\{\gamma_n\}$ in (\ref{215}) is simply
\be
\gamma_n=\overline{[n+1]_q},
\label{427}\en
$\forall n\geq0$. The other results in (\ref{217})-(\ref{221}) can be easily checked. We just want to observe that the coefficients in (\ref{220}) and (\ref{221}) can be simplified, in our situation, since
$$
\gamma_n-\gamma_{n-1}=q^n,
$$
$\forall n\geq1$, and the only case in qhich this do not depend on $n$ is when $q=1$, i.e. for pseudo-bosons, \cite{bagspringerbook}.

We conclude this section by noticing that not much can be said on bi-coherent states since we are here not in a position to check if bounds as those in (\ref{38}) are satisfied or not. This could be done only when replacing our abstract operators $a$ and $b$ with some explicit realization in some concrete Hilbert space.

\section{Another example: the deformed Heisenberg algebra}\label{sectDGHA}

In this section we will describe in details another class of examples fitting our construction in Section \ref{sect2}. This class arises from what has been called DGHA in \cite{bcg}, and which is essentially a (minimal) set of rules relating three operators which give the possibility to find eigenvalues and eigenvectors of one of the three operators involved.
In particular, after a short review of the results in \cite{bcg}, restated in a way which is more convenient for us, we will show how each DGHA can be seen as a special case of what described in Section 2.

Suppose we have three operators $h$, $a$, $b$, acting on the Hilbert space $\Hil$ and densely defined on a dense subspace $\D\subseteq\Hil$, stable under the action of these operators and of their adjoints. Let $f(x)$ be a real and strictly increasing function. 

\begin{defn}\label{defALO}
We say that 	$h$, $a$, $b$ and $f(x)$ obey a DGHA if the following identities are satisfied:
\be
h\,b=b\,f(h), \qquad a\,h=f(h)\,a,\qquad [a,b]=f(h)-h.
\label{51}\en 
\end{defn}
We refer to \cite{bcg} for a discussion on the definition of $f(h)$ if $h$ is unbounded and not self-adjoint. Here we will always assume that $f(h)$ is well defined and leaves $\D$ invariant, together with its adjoint $(f(h))^\dagger=f(h^\dagger)$.

\vspace{2mm}

{\bf Remark:--} We could rather assume, as already stressed several times in this paper, that all the operators involved in our construction, including $f(h)$, belong to $\Lc^\dagger(\D)$, for some $\D$, see Appendix 1.

\vspace{2mm}

Let us assume, \cite{bcg}, that two non zero vectors $\xi_0$ and $\eta_0$ exist in $\D$ such that
\be
h\xi_0=0,\qquad h^\dagger \eta_0=0.
\label{52}\en
This means that zero is an eigenvalue of $h$ and $h^\dagger$, which we are considering to be different, in general. We define the vectors
\be
\xi_n=\frac{1}{\sqrt{\epsilon_n!}}\, b^n\,\xi_0, \qquad \eta_n=\frac{1}{\sqrt{\epsilon_n!}}\, {a^\dagger}^n\,\eta_0,
\label{53}\en
$n\geq0$, where $\epsilon_0=0$ and $\epsilon_n=f(\epsilon_{n-1})$, $n\geq1$. Of course, this definition and the fact that $f(x)$ is strictly increasing, imply that
\be
0=\epsilon_0<\epsilon_1<\epsilon_2<\ldots.
\label{54}\en
In (\ref{53}) we have put,
as usual, $\epsilon_0!=1$ and $\epsilon_n!=\epsilon_1\epsilon_2\cdots\epsilon_n$, $n\geq1$. Notice that the factor $\sqrt{\epsilon_n!}$ is real, in view of what we have just discussed. Since $\D$ is stable under the action of $b$ and $a^\dagger$, it follows that $\xi_n,\eta_n\in\D$, for all $n\geq0$. The following results can be easily deduced, see also \cite{bcg}:
\begin{enumerate}
	
	\item $\xi_n$ and $\eta_n$ are respectively eigenstates of $h$ and $h^\dagger$:
\be
h\xi_n=\epsilon_n\xi_n, \qquad h^\dagger\eta_n=\epsilon_n\eta_n, 
\label{55}\en
$\forall n\geq0$. Then, since each eigenvalue is non degenerate,
\be
\langle\xi_n,\eta_m\rangle=0,
\label{56}\en
if $n\neq m$.

\item If $\F_\eta=\{\eta_n, \,n\geq0\}$ is complete in $\Hil$, then $a\xi_0=0$. Also, if $\F_\xi=\{\xi_n, \,n\geq0\}$ is complete in $\Hil$, then $b^\dagger\eta_0=0$.

\item The following ladder equations are satisfied:
\be
b\xi_n=\sqrt{\epsilon_{n+1}}\,\xi_{n+1}, \qquad a^\dagger\eta_n=\sqrt{\epsilon_{n+1}}\,\eta_{n+1},
\label{57}\en
as well as
\be
a\xi_n=\sqrt{\epsilon_{n}}\,\xi_{n-1}, \qquad b^\dagger\eta_n=\sqrt{\epsilon_{n}}\,\eta_{n-1},
\label{58}\en
$n\geq1$, together with $a\xi_0=b^\dagger\eta_0=0$, see item 2 in this list.
\item The following eigenvalue equations are satisfied:
\be
b\,a\xi_n=\epsilon_n\xi_n, \qquad a\,b\xi_n=\epsilon_{n+1}\xi_n, \qquad a^\dagger\,b^\dagger\eta_n=\epsilon_{n}\eta_n,\qquad b^\dagger\,a^\dagger\eta_n=\epsilon_{n+1}\eta_n,
\label{59}\en
$\forall n\geq0$.
\item if we choose the normalization of $\xi_0$ and $\eta_0$ such that $\langle\xi_0,\eta_0\rangle=1$, then
\be
\langle\xi_n,\eta_m\rangle=\delta_{n,m}
\label{510}\en
$\forall n,m\geq0$. Notice that this equation {\em refines} the one in (\ref{56}).

\end{enumerate}

We refer to \cite{bcg} for some examples of DGHA in a quantum mechanical context. In particular, a deformed P\"oschl-Teller and an (again, deformed) infinite square-well potentials are considered.  

\vspace{2mm}

We will now show how each DGHA fits in what discussed in Section \ref{sect2}. In particular, we will show that $a$ and $b$ in Definition \ref{defALO} are indeed ALOs.

To show this fact, we first check that taken $a$ and $b$ and $h$ in Definition \ref{defALO},  $(h,a,b)\in\Rc_1$. Indeed we have, using twice (\ref{51}),
$$
[h,b]=hb-bh=bf(h)-bh=b(f(h)-h)=b[a,b],
$$
which is (\ref{21}) with the identification $(H,T,S)\leftrightarrow(h,a,b)$, and $\lambda=1$. This means that Proposition \ref{prop1} can be applied. Moreover, (\ref{22}) is also satisfied. In fact we have, using again (\ref{51}),
$$
[H,[T,S]] \longrightarrow [h,[a,b]]=[h,f(h)-h]=0,
$$
as required in (\ref{22}). The assumptions of Theorem \ref{th1} are also satisfied:  $(h,a,b)\in\Rc_1$, the eigenvalues of $h$, the $\epsilon_n$'s, are all non degenerate, and a non zero vector $\varphi_0=\xi_0$ exists such that $h\varphi_0=0$, see (\ref{52}).

The vector $\varphi_n$ in (\ref{23}) is $\varphi_n=b^n\xi_0=\sqrt{\epsilon_n!}\,\xi_n$, see (\ref{53}). The coefficient $\mu_n$ in (\ref{24}) can be easily computed:
$$
[T,S]\varphi_n\longrightarrow [a,b]\left(\sqrt{\epsilon_n!}\,\xi_n\right)=\sqrt{\epsilon_n!}(f(h)-h)\xi_n=(\epsilon_{n+1}-\epsilon_n)\left(\sqrt{\epsilon_n!}\,\xi_n\right),
$$
where we have also used that $f(\epsilon_n)=\epsilon_{n+1}$. Hence $\mu_n=\epsilon_{n+1}-\epsilon_n$, $\forall n\geq0$. Using this expression in (\ref{26}), and recalling that $\lambda=1$, we get $E_n=\sum_{k=0}^{n-1}\mu_k=\epsilon_n$, in agreement\footnote{Since $\xi_n$ and $\varphi_n$ differ only for a normalization, $\xi_n$ is an eigenstate of $h$ with eigenvalue $\epsilon_n$ if and only if $\varphi_n$ is an eigenstate of $h$ with the same eigenvalue.} with (\ref{55}). 

If $\F_\varphi$ is a basis, it admits an unique biorthonormal basis, \cite{chri,heil}, $\F_\psi=\{\psi_n\}$. In view of (\ref{510}), its vectors can be easily identified: $\psi_n=\frac{1}{\sqrt{\epsilon_n!}}\,\eta_n$, $\forall n\geq0$, and it is clear that these are eigenvectors of $h^\dagger$, with eigenvalue $\epsilon_n$: $h^\dagger\psi_n=\epsilon_n\psi_n$, in agreement with (\ref{28}). The lowering equations in (\ref{29}) are also easily checked. For instance, if $n\geq1$,
$$
S^\dagger\psi_n\longrightarrow b^\dagger \left(\frac{1}{\sqrt{\epsilon_n!}}\,\eta_n\right)=\frac{1}{\sqrt{\epsilon_{n-1}!}}\,\eta_{n-1}=\psi_{n-1},
$$
as expected. In this derivation we have used (\ref{58}). As for the equality in (\ref{212}) we can check that both sides are indeed equal to $\mu_n$. What is more relevant, is to check if $(h,a,b)$ also belongs to $\Rc_1^{(s)}$, in the sense of Definition \ref{def2}. This is true, in fact:
$$
[h,a]=ha-ah=ha-f(h)a=(h-f(h))a=[b,a]a,
$$ 
using twice (\ref{51}). As a consequence of this property, the last part of Section \ref{sect2} holds for DGHA. In particular, the coefficient $\gamma_n$ in (\ref{215}) can be identified:
$$
T^\dagger\psi_n\longrightarrow a^\dagger \left(\frac{1}{\sqrt{\epsilon_n!}}\,\eta_n\right)=\frac{1}{\sqrt{\epsilon_n!}}\,\sqrt{\epsilon_{n+1}}\,\eta_{n+1}=\epsilon_{n+1}\psi_{n+1},
$$
using (\ref{57}). Hence $\gamma_n=\epsilon_{n+1}$, $\forall n\geq0$, and formulas (\ref{216})-(\ref{219}) easily follow. 

For our DGHA something can be deduced for the difference $\gamma_n-\gamma_{n-1}$ in (\ref{220}) and (\ref{221}). Indeed we have $\gamma_n-\gamma_{n-1}=\epsilon_{n+1}-\epsilon_n=f(\epsilon_n)-\epsilon_n$, which is independent of $n$ if $f(x)=x+k$, $k\in\mathbb{R}$, but not with other choices of $f(x)$. This is not surprising, since  this expression of $f(x)$ implies that $[a,b]=f(h)-h=k\1$, which means that, for instance, $\frac{1}{k}a$ and $b$ are pseudo-bosonic operators, \cite{bagspringerbook}, for which all the results deduced here are well known.

As already stressed for pseudo-quons, we notice that not much can be said on bi-coherent states since we are here not in a position to check if bounds as those in (\ref{38}) are satisfied or not. This could be done only when replacing our abstract operators $a$, $b$ and $h$, and the function $f$, with some explicit realization in some concrete Hilbert space.

\section{Conclusions}\label{sectconcl}

We have discussed here the role of ALOs in the analysis of some physical system driven by a non self-adjoint Hamiltonian.

In particular, also in view of the examples discussed in Sections \ref{sectDQ} and \ref{sectDGHA}, the approach proposed here appears to be rather general, since it covers these examples, and all those which were originally considered in \cite{bagALO,bagquons,bcg}, as well as many others in which the pseudo-quons are replaced by ordinary quons or by pseudo-bosons, which are both particular cases of the pseudo-quons considered in Section \ref{sect2}. Of course, looking for more applications is quite interesting and, in view of the generality of the method, we hope to find other interesting examples which can be algebraically treated in terms of our ALOs.

It is clear that our analysis of bi-coherent states here is just preliminary. We plan to go back soon to these states with concrete examples, and possibly with physical applications. 

We should also mention that a side result in our analysis consists in the possibility of extending $q$-onic commutation relations, see (\ref{41}), to complex values of $q$. As we have commented, this is not possible for ordinary quons. Also, pseudo-quons can be used to diagonalize a non self-adjoint oscillator-like Hamiltonian, see (\ref{415}), involving position and momentum operators  satisfying a commutation rule, see (\ref{414}), of the kind one meets when dealing with minimal length quantum mechanics. We believe that this aspect deserves a  deeper analysis, too, in a close future, also in connection with the consequences for uncertainty relations.

\section*{Acknowledgements}

The author acknowledges partial support from Palermo University and from G.N.F.M. of the INdAM, and from PRIN Project {\em  Transport phenomena in low
	dimensional structures: models, simulations and theoretical
aspects}.

\renewcommand{\theequation}{A.\arabic{equation}}

\section*{Appendix 1: $O^*$-algebras}\label{appendix}

Let us briefly review how  $\Lc^\dagger(\D)$ can be introduced, and why it is so relevant for us.  We refer to \cite{aitbook,trrev,bagrev2007} for many results on $*$-algebra, { quasi $*$-algebras}, and $O^*$-algebras. In particular, we have:

\begin{defn}\label{o*}Let $\mathcal{H}$ be a separable Hilbert space and $N_0$ an
	unbounded, densely defined, self-adjoint operator. Let $D(N_0^k)$ be
	the domain of the operator $N_0^k$, $k \ge 0$, and $\mathcal{D}$ the domain of
	all the powers of $N_0$, that is,  $$ \mathcal{D} = D^\infty(N_0) = \bigcap_{k\geq 0}
	D(N_0^k). $$ This set is dense in $\mathcal{H}$. We call
	$\mathcal{L}^\dagger(\mathcal{D})$ the $*$-algebra of all \textit{  closable operators}
	defined on $\mathcal{D}$ which, together with their adjoints, map $\mathcal{D}$ into
	itself. Here the adjoint of $X\in\mathcal{L}^\dagger(\mathcal{D})$ is
	$X^\dagger=X^*_{| \mathcal{D}}$. $\mathcal{L}^\dagger(\mathcal{D})$ is called  an $O^*$-algebra.
\end{defn}

In $\mathcal{D}$ the topology is defined by the following $N_0$-depending
seminorms: $$\phi \in \mathcal{D} \rightarrow \|\phi\|_n\equiv \|N_0^n\phi\|,$$
where $n \ge 0$, and  the topology $\tau_0$ in $\mathcal{L}^\dagger(\mathcal{D})$ is introduced by the seminorms
$$ X\in \mathcal{L}^\dagger(\mathcal{D}) \rightarrow \|X\|^{f,k} \equiv
\max\left\{\|f(N_0)XN_0^k\|,\|N_0^kXf(N_0)\|\right\},$$ where
$k \ge 0$ and   $f \in \mathcal{C}$, the set of all the positive,
bounded and continuous functions  on $\mathbb{R}_+$, which are
decreasing faster than any inverse power of $x$:
$\mathcal{L}^\dagger(\mathcal{D})[\tau_0]$ is a {   complete *-algebra}.

The relevant aspect of $\LD$ is that, \cite{aitbook,trrev,bagrev2007}, if $x,y\in \mathcal{L}^\dagger(\mathcal{D})$, we can multiply them and the results, $xy$ and $yx$, both belong to $\mathcal{L}^\dagger(\mathcal{D})$, as well as their difference, the commutator $[x,y]$. Also, powers of $x$ and $y$ all belong to $\Lc^\dagger(\D)$, which is therefore a good framework to work with, also in presence of unbounded operators. For instance, if $N_0=a^\dagger a$, where $[a,a^\dagger]=\1$, we can prove that $a, a^\dagger\in\Lc^\dagger(\D)$. Hence $N_0\in  \Lc^\dagger(\D)$ as well. This is also true for pseudo-bosonic operators, \cite{bagrus2018}, at least if $\D=\Sc(\mathbb{R})$.

\renewcommand{\theequation}{A.\arabic{equation}}

\section*{Appendix 2: Graphene}\label{appendix2}

The Hamiltonian for the two Dirac points $K$ and $K'$ can be written as, see \cite{geim},
\be
H_D=\left(
\begin{array}{cc}
	H_K & 0 \\
	0 & H_{K'} \\
\end{array}
\right),
\label{b1}\en
with
\be
H_K=v_F\left(
\begin{array}{cc}
	0 & p_x-ip_y+\frac{eB}{2}(y+ix) \\
	p_x+ip_y+\frac{eB}{2}(y-ix) & 0 \\
\end{array}
\right),
\label{b2}
\en
and $H_{K'}=H_K^T$. Here $x,y,p_x$ and $p_y$ are the usual self-adjoint, two-dimensional position and momentum operators: $[x,p_x]=[y,p_y]=i\1$, all the other commutators being zero.  $\1$ is the identity operator in the Hilbert space $\Kc=\Lc^2(\Bbb R^2)$, whose scalar product is indicated as $\pin{.}{.}$. The factor $v_F$ is the Fermi velocity.

We then put $\xi=\sqrt{\frac{2}{eB}}$, 
$
X=\frac{1}{\xi}x$, $Y=\frac{1}{\xi}y$, $P_X=\xi p_x$,  $P_Y=\xi p_y$, and then
 $a_X=\frac{X+iP_X}{\sqrt{2}}$, $a_Y=\frac{Y+iP_Y}{\sqrt{2}}$, and still
\be
A_1=\frac{a_X-ia_Y}{\sqrt{2}},\qquad A_2=\frac{a_X+ia_Y}{\sqrt{2}}.
\label{b3}\en
The following commutation rules are satisfied:
\be
[a_X,a_X^\dagger]=[a_Y,a_Y^\dagger]=[A_1,A_1^\dagger]=[A_2,A_2^\dagger]=\1,
\label{b4}
\en
the other commutators being zero. Then we have:
\be
H_K=\frac{2iv_F}{\xi}\left(
\begin{array}{cc}
	0 & A_2^\dagger \\
	-A_2 & 0 \\
\end{array}
\right),
\label{b5}\en
which is manifestly Hermitian: $H_K=H_K^\dagger$. Since $H_K$ does not depend on $A_1$ and $A_1^\dagger$, its eigenstates must be degenerate. Let $e_{0,0}\in \Kc$ be the non zero vacuum of $A_1$ and $A_2$: $A_1e_{0,0}=A_2e_{0,0}=0$. Then we introduce
\be
e_{n_1,n_2}=\frac{1}{\sqrt{n_1!n_2!}}(A_1^\dagger)^{n_1}(A_2^\dagger)^{n_2}e_{0,0},
\label{36}
\en
and the set $\EE=\left\{e_{n_1,n_2},\, n_j\geq0\right\}$. $\EE$ is an o.n. basis for $\Kc$.  However, to deal with $H_K$, it is convenient to work in $\Kc_2=\Kc\oplus\Kc$, the direct sum of $\Kc$ with itself:
$$
\Kc_2=\left\{f=\left(
\begin{array}{c}
	f_1 \\
	f_2 \\
\end{array}
\right),\quad f_1,f_2\in\Kc
\right\}.
$$
In $\Kc_2$ the scalar product $\pin{.}{.}_2$ is defined as usual:
\be
\pin{f}{g}_2:=\pin{f_1}{g_1}+\pin{f_2}{g_2},
\label{37}\en
and the square norm is $\|f\|_2^2=\|f_1\|^2+\|f_2\|^2$, for all $f=\left(
\begin{array}{c}
	f_1 \\
	f_2 \\
\end{array}
\right)$, $g=\left(
\begin{array}{c}
	g_1 \\
	g_2 \\
\end{array}
\right)$ in $\Kc_2$. We now introduce the vectors
\be
v_{n_1,0}=\left(
\begin{array}{c}
	e_{n_1,0} \\
	0 \\
\end{array}
\right), \qquad \mbox{and}\qquad
v_{n_1,n_2}^{(\pm)}=\frac{1}{\sqrt{2}}\left(
\begin{array}{c}
	e_{n_1,n_2} \\
	\mp i e_{n_1,n_2-1} \\
\end{array}
\right), \quad \mbox{if } n_2\geq1,
\label{b10}\en
$\forall n_1\geq0$.
Now, we call
$\V_2=\{v_{n_1,n_2}^{(k)},\,n_1\geq0,\,n_2\geq1,\,k=\pm\}\cup\{v_{n_1,0},\,n_1\geq0\}$.
It is easy to check that these vectors are mutually orthogonal, normalized in $\Kc_2$, and total. Hence, $\V_2$ is an o.n. basis for $\Kc_2$. Its vectors are eigenvectors of $H_K$:
\be
H_K v_{n_1,0}=0, \quad H_K v_{n_1,n_2}^{(+)}=E_{n_1,n_2}^{(+)}  v_{n_1,n_2}^{(+)}, \quad   H_K v_{n_1,n_2}^{(-)}=E_{n_1,n_2}^{(-)}  v_{n_1,n_2}^{(-)},
\label{b11}\en
where $E_{n_1,n_2}^{(\pm)}=\pm \frac{2v_F}{\xi}\,\sqrt{n_2}$. This shows that the discrete spectrum of $H_K$ is not bounded from below (neither from above). For this reason, it is impossible to shift the Hamiltonian to get a new operator $H_K+\gamma\1_2$ (here $\1_2$ is the identity operator on $\Kc_2$) with a purely positive set of eigenvalues.

\end{document}